\documentclass[9pt]{article}

\usepackage{amsthm, amsfonts, amssymb, stmaryrd, amsmath}
\usepackage{epsfig}
\usepackage{footmisc}

\def\R{\mathbb R}
\def\Q{\mathbb Q}

\def\S{$\Sigma_g$\ }

\def\SS{\Sigma_g\ }
\def\SS1{\Sigma_1\ }

\newenvironment{proof-sketch}{\\ \indent{\it Proof sketch.}\hspace*{1em}}{\hbox{}\hfill$\Box$\bigskip}

\newtheorem{defin}{Definition}[section]
\newtheorem{alg}[defin]{Algorithm}
\newtheorem{thm}[defin]{Theorem}

\newtheorem{lem}[defin]{Lemma}
\newtheorem{prop}[defin]{Proposition}

\newtheorem{remark}[defin]{remark}

\title{Local Routing in Graphs Embedded on Surfaces of Arbitrary Genus}
\author{Maia Fraser\thanks{Department of Computer Science, University of Chicago, 1100 East 58th Street,
Chicago, IL 60637 ({\tt maia@cs.uchicago.edu}).}}
\date{November 11, 2010}
       
\begin{document}

\maketitle

\begin{center} 
\end{center}

\pagestyle{myheadings}
\thispagestyle{plain}
\markboth{Maia Fraser}
{Local Routing on Surfaces of Arbitrary Genus}

\begin{abstract}
We present a local routing algorithm which guarantees delivery in all connected graphs embedded on a known surface of genus $g$. The algorithm transports $O(g\log n)$ memory and finishes in time $O(g^2n^2)$, where $n$ is the size of the graph. It requires access to a homology basis for the surface. This algorithm, GFR, may be viewed as a suitable generalization of Face Routing (FR), the well-known algorithm for plane graphs, which we previously showed does {\it not} guarantee delivery in graphs embedded on positive genus surfaces. The problem for such surfaces is the potential presence of homologically non-trivial closed walks which may be traversed by the right-hand rule. We use an interesting mathematical property of homology bases (proven in Lemma~\ref{lem:connectFaceBdr}) to show that such walks will not impede GFR. FR is at the base of most routing algorithms used in modern (2D) ad hoc networks: these algorithms all involve additional local techniques to deal with edge-crossings so FR may be applied. GFR should be viewed in the same light, as a base algorithm which could for example be tailored to sensor networks on surfaces in 3D. Currently there are no known efficient local, logarithmic memory algorithms for 3D ad hoc networks. From a theoretical point of view our work suggests that the efficiency advantages from which FR benefits are related to the codimension~one nature of an embedded graph in a surface rather than the flatness of that surface (planarity).


\end{abstract}

\section{Introduction}\label{intro}

Face routing (FR) was introduced by Kranakis, Urrutia et al in 1999 in \cite{kranurr}. It guarantees delivery in time linear in the size of the network, using logarithmic memory, and was the first algorithm to exploit availability of position, now a common feature in many ad hoc networks. Its impact was twofold: it showed that it was possible to find {\it local} algorithms for ad hoc networks, as opposed to constantly re-adjusting routing tables, and it revealed the important role {\it position} could play in achieving this. The word ``local" in this context means that the algorithm is executed by a mobile agent traveling through the network and using only local information at each node (including position) together with a ``small" amount of transported memory (commonly restricted to be logarithmic in the size of the network). 

More recently, in 2004, work of Reingold \cite{stconn} showed universal exploration sequences (UXS's) are constructible in logspace thus implying that it is always theoretically possible to route locally in the above sense, even without position data (transporting memory of order 
$O(\log n)$ as overhead in the messages)\footnote{The author thanks Mark Braverman for his comments on this. See also \cite{braverm}.}. However delivery time by this method is a high polynomial in the network size (not yet optimized), making it impractical for most ad hoc applications. This makes the critical challenge for routing in 3D ad hoc networks  {\it time efficiency} and it opens the question of whether position on a surface can help improve this efficiency in 3D networks as it did in the 2D setting. We answer this question in the affirmative.

Related recent work shows that, even using geometric information, no constant 
memory local routing algorithm exists for 3D UBW's that occupy a slab of thickness greater than $\frac1 {\sqrt 2}$ (see Durocher et al \cite{lata}), however that work does {\it not} address algorithms using logarithmic memory. 

Another related result shows that no {\it sublinear}
memory local routing algorithm exists for directed graphs, even if they are planar (Fraser et al \cite{fraJoint}). This means no local routing algorithm exists at all in ad hoc networks in the presence of uni-directional communication links, under the common (logarithmic-memory) definition of ``local" used in the ad hoc community.

In this paper we make a first step towards addressing local, logarithmic-memory routing algorithms for undirected 3D ad hoc networks, by considering those which lie on {\it surfaces}. In particular we are interested in ad hoc networks (for example of sensors) which may be deployed on terrains and 
outer surfaces of physical structures such as vehicles, buildings or other objects. We assume such surfaces are in fact oriented two-dimensional smooth manifolds which we will approximate by piecewise linear surfaces (PLS's). We discuss this in more detail in the next section. Given a PLS (and a homology basis for it) there exists a local, logarithmic-memory, quadratic-time routing algorithm, called {\it generalized Face Routing} (GFR), which is correct for {\it all} connected graphs that are ``reasonably" embedded on the surface. 

By ``reasonable" we mean such that any two edges and/or reference curves (connecting curve, homology curves) may intersect each other only a bounded number of times (independent of graph size). We remark that Face Routing in the plane assumes edges of the graph to be straight lines and assumes a linear or piecewise linear connecting curve between sender and destination, thus achieving this intersection property. Similarly in our case we take piecewise linear (PL) curves to guarantee reasonableness.  


Like Face Routing, GFR relies on the absence of edge-crossings. It is therefore, like Face Routing, a base algorithm which would require further processing to handle edge-crossings and make it suitable for real-world ad hoc networks (see \cite{kranurr2, kuhn}). With appropriate assumptions on the network and surface it is possible to apply any of the standard methods for handling edge-crossings in the plane \cite{kranurr2, kuhn, lata} to surfaces, however the practicality of the assumptions involved would have to be studied in a more applied paper. We welcome input from readers with specific applications in mind.

\section{Terminology and Assumptions}

Assume \S is a piecewise linear surface (PLS) of genus $g$ that has been triangulated. This is a natural assumption in practice (see Section~\ref{sec:model} below). 
We recall the {\bf genus} of a surface is defined as the maximal 
number of disjoint, simple closed curves which may be removed from the surface without disconnecting it. 

The algorithm GFR will be executed by an agent traveling in an embedded graph 
$G$ on $\Sigma_g.$ We assume, as is typical in ad hoc networks, that the time of transmission is small compared with the speed at which the network changes
so that one may reasonably assume the network is fixed but unknown during execution. The agent does not have access to a global representation of the graph, which may be constantly changing, but only to the ID's of its current and neighboring vertices as well their positions on $\Sigma_g.$ In addition it has access to a {\it homology basis} $\mathcal B$ 
for $\Sigma_g.$ This will consist of closed simple curves contained in the $1$-dimensional simplicial subcomplex of \S.

\subsection{Homology and homology bases} \label{sec:hom} The reader who is not familiar with the concept of homology or homology basis is referred to Jeff Erickson's \cite{erickson} for an introduction to these concepts from a computational geometry perspective or to Hatcher \cite{hatcher} for a more detailed mathematical treatment. 
We give only a minimal overview here, of {\bf simplicial homology} for surfaces. Assume a ring $R$. In this paper it will be $\Q$. For $k =0, 1, 2$, 
let $C_k(\Sigma_g;R)$ be the set of all formal linear combinations of $k$-simplices in \S with coefficients in $R$. $C_k(\Sigma_g;R)$ is naturally a group under addition; its elements are called {\bf $k$-chains}. Define the {\bf boundary operator} 
$$\delta_k : C_k(\Sigma_g;R) \to C_{k-1}(\Sigma_g;R)$$ 
to be the linear map which sends each oriented simplex to the sum of
of its oriented boundary simplices (or to zero if there are none). 
Those $k$-chains which map to zero (i.e. have no boundary) are called {\bf k-cycles}. 
They form a group which we denote
$Z_k = \ker \delta_k$ for $k = 1, 2$. For our case of $R = \Q$, connected $1$-cycles with integer coefficients can be interpreted as closed unparametrized curves which are contained in the $1$-dimensional simplicial sub-complex of \S. Unconnected $1$-cycles with integer coefficients are formal sums of connected ones. Let $B_k = \delta_{k+1}(C_{k+1})$ denote the $k$-boundaries, i.e. images of the boundary operator $\delta_{k+1}$ for $k = 0, 1$.
Define the first homology group $H_1(\Sigma_g;R)$ to be the quotient group 
$Z_1/B_1$. Its elements are called {\bf homology classes}. They are the equivalence classes for the equivalence relation ``homologous to": two $1$-cycles $\beta$ and $\beta'$ are said to be {\bf homologous} if their difference is a $1$-boundary, i.e. $\beta$ and $-\beta'$ are the two boundary components of some formal sum of $2$-simplices. When $R = \Q$ (or any other field), $H_1(\Sigma_g;R)$ is in fact a vector space.  A {\bf homology basis} is technically a basis of this vector space and so consists of homology classes. But, in fact, one can always find a set of $1$-cycles which are connected and have integer coefficients whose homology classes form a homology basis. As a result, it is common in the computational geometry literature to speak of the set of these {\it closed curves}  (instead of their homology classes) as the {homology basis}. We will adopt this usage. We remark that $H_1(\Sigma_g; \Q) = \Q^{2g}$ so a homology basis for \S will consist of $2g$ closed curves.

\subsection{Closed curves} As mentioned above, we may interpret $1$-cycles with integer coefficients as closed (oriented) unparametrized curves or linear combinations thereof. Moreover {\it all} closed curves used in the homological arguments of this paper in Section~\ref{sec:proof} will be oriented and simplicial, i.e. are actually connected $1$-cycles with integer coefficients. Therefore, in Section~\ref{sec:proof} we will write {\bf closed curve} as a synonym for ``connected $1$-cycle with integer coefficients", and will write {\bf simple closed curve} if the curve has no repeated vertices\footnote{The latter would unfortunately be a ``cycle" in standard graph theoretic terminology; our convention avoids this.}. 
Finally, when a parametrization is assumed, we will specify this and assume the orientation is compatible. See also Remark~\ref{rem:walk}.

\subsection{Planar Representation} \label{sec:planar}
A set of disjoint simple closed curves $\mu_1, \ldots, \mu_g$ on the surface $\Sigma_g$ which do not disconnect \S may always be completed to a homology basis 
$$\mathcal B = \{ \mu_1, \ldots, \mu_g, \lambda_1, \ldots, \lambda_g \}.$$ 
The curves $\mu_1, \ldots, \mu_g$ also
give a convenient way of representing the surface in the plane (see Fig.~\ref{fig:handles})\footnote{For illustrations in the present paper we will additionally assume that $\mu_i \cap \lambda_j$ is a single point if $i=j$ and empty otherwise (such a basis can always be found), but this is not necessary to the algorithm; any homology basis will do. }. Indeed, let $\mathcal C$ consist of the union of these curves. After removing $\mathcal C$ we know the resulting surface is homeomorphic to a multiply punctured plane,
namely some $\mathcal U = \R^2 \setminus \bigcup\limits_{i=1}^g D_i$, where the $D_i$ are disks. Once we assume a particular
homeomorphism $\Phi$ from $\Sigma_g \setminus \mathcal C$ onto $\mathcal U \subset \R^2$,
we can map all points and curves of interest in $\Sigma_g$ (that are not completely contained in the special $\mathcal C$) to corresponding points and curves in the plane. We call this a {\bf planar representation} of the surface. 

Moreover, an embedded graph $G$ on \S gives rise to $\Phi(G)$, an embedded object (points and arcs) on $\mathcal U \subset \R^2$, which we can interpret as a graph after making the identifications given by $\mathcal C$: edges will be unions of curve segments with breaks occurring at elements of $\mathcal C$. As the agent for GFR travels in $G$, all coordinates of neighbors it obtains will be given in terms of planar coordinates in $\mathcal U$. See Section~\ref{sec:data} for more details. 

We remark however that the proof of correctness that we give for GFR in Section~\ref{sec:proof} is not on the plane but rather on the PLS surface \S, as this view is notationally simpler for the homological arguments involved.

\begin{figure}[h] 
\centerline{\epsfig{file=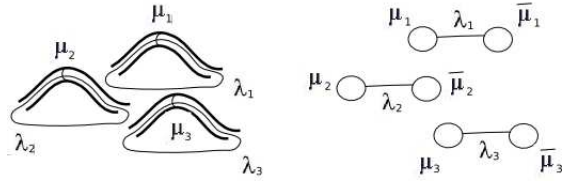,height=2.5cm}}
\medskip
\caption{To the left: surface $\Sigma_g$ of genus $g = 3$ viewed as the plane with $3$ handles attached (after removal of one point). To the right: \S viewed as the plane with $6$ holes pairwise connected by the $\lambda_i$.
The set 
$\{\mu_i, \lambda_i, : 1 \leq i \leq 3\}$ is a homology basis for $\Sigma_g$.} \label{fig:handles}
\end{figure}

\subsection{Note on Practical 3D Model and Future Work} \label{sec:model}
The property that \S be a triangulated PLS arises naturally from the standard modeling schemes for representing real surfaces. We give two examples.

\paragraph{Triangular Mesh Surface Models} This is one of the most common modeling schemes for representing an idealized surface\footnote{(a triangular polytope)} in computer graphics. Such a presentation consists primarily of a set of triangles (given by their vertex coordinates in $\R^3$) and may also organize this information (e.g. in so-called ``triangle strips") to facilitate stepping systematically from each triangle to a neighbor. Triangular Mesh Surfaces are in fact triangulated PLS's. Suppose we have fixed a three dimensional ambient coordinate system in which the real world object of interest (e.g. the surface of a building) is situated. Efficient algorithms exist for generating a PLS which approximates a smooth surface using samples from the surface (see for example \cite{margaliot}). One then triangulates the linear pieces if they are not already triangular so as to obtain a triangular mesh model. 
To take this a step further, suppose one has
two triangular mesh surface models, $\hat\Sigma_{\mbox{\tiny IN}}$ and $\hat\Sigma_{\mbox{\tiny OUT}}$, one on either side of the actual surface  $\Sigma_g$ and suitably close together. It is relatively straightforward to 
construct a piecewise linear map $\Phi$ from 
either $\hat\Sigma_{\mbox{\tiny IN}}$ or $\hat\Sigma_{\mbox{\tiny OUT}}$ to the plane. If these two surfaces are built ``compatibly" then one can uniquely map points lying between the two surfaces to points on the plane. This allows one to convert three dimensional coordinates of points on the real world \S to coordinates of points on the plane. Also, straight lines connecting nearby points on the three dimensional surface \S would be mapped to PL curves in the plane and one could then use standard planar techniques to deal with edge-crossings. The aspects of ``compatibility" one would require of the triangle mesh models would depend on which edge-crossing techniques one wishes to use and also the particular setting. We welcome comments from people with particular sensor network settings in mind.

\paragraph{Wire Frame Models} Note that in structural analysis, wire frame models are often the primary objects used to model curved surfaces and an appropriate triangular mesh representation is then generated to fill out the wire frame model (see \cite{klein, yamada}). For simplicity, the union of homology curves $\mathcal C_{\mbox{\tiny IN}}$ (resp. $\mathcal C_{\mbox{\tiny OUT}}$) could be taken to be contained in the original wire frame model.

\subsection{Data required by GFR} \label{sec:data}

We now summarize GFR's access to data. Each node has stored the model \S (for example in form of a triangular mesh model) and the instructions for computing $\Phi$ (see Section~\ref{sec:planar}). The model \S may be very approximate and the corresponding data small, depending on the application. Its main purpose is to record homological information\footnote{There exist various manifold reconstruction techniques from topological data analysis in recent years which yield good approximate manifolds with coinciding homological invariants.}. It is moreover constant: it depends {\it only} on the surface and not on the number of nodes or their positions. Included in this data is the homology basis $\mathcal B$. For example, assuming a triangular mesh model, this can be done by recording for every $1$-simplex whether or not it belongs to the $i$'th curve of $\mathcal B$ (for $i \in [2g]$) and its position in a sequential numbering of the $1$-simplices of that curve.

Each node can (via GPS) obtain its momentary latitude, longitude and altitude. From these it computes (using $\Phi$) its {\bf planar position} on $\mathcal U \subset \R^2$. The agent executing GFR can, at each node, query the neighboring nodes for their planar positions. 

When the agent is started by a sender $S$ wishing to transmit to a destination $T$, it computes a {\bf connecting curve} $\gamma$. This is any PL curve in the plane connecting $\Phi(S)$ to $\Phi(T)$. $\gamma$ may be mostly composed of $1$-simplices, to simplify calculations. 

GFR therefore operates essentially {\it in the plane} (within $\mathcal U \subset \R^2$) but uses identifications of certain curves in the plane and other homological data to accomplish its task.

\section{Generalized Face Routing} \label{gfr}

\subsection{Basic Definitions and Face Routing} \label{basicfr}

\begin{defin} [Border Walk]
A walk $\beta$ in $G$ which has no edges on its left side will be called a \emph{border walk}. 
\end{defin}

Note that border walks are boundary components of {\bf regions}, where a region is a connected component of the complement of $G$ in $\Sigma_g.$ In general, regions in \S may have multiple boundary components. We assume the standard boundary orientation convention: a boundary curve $\beta$ is oriented so that the region it bounds lies to the left when $\beta$ is traversed in its preferred sense. Clearly each {\it directed} edge forms part of the boundary of exactly one region. We thus define:

\begin{defin} [Adjacent Border Walk]
The (at most) two border walks which share a given (undirected) edge are said to be \emph{adjacent} at that edge.
\end{defin}

\begin{defin} [Trivial or Non Trivial Border Walks (NTBW), Tiled Region]
If a border walk $\beta$ is the only boundary component of some region, we say $\beta$ is a \emph{trivial border walk}, otherwise it is a \emph{non-trivial border walk} (abbreviated NTBW). In other words, a closed curve is an NTBW if and only if it is not homologous to zero. We define the \emph{tiled region}\footnote{We remark that we are departing from the definition of $\mathcal R(G)$ in earlier papers by the author by here including in $\mathcal R(G)$ any region with a connected boundary, regardless of its genus.} $\mathcal R(G)$ to be the union of $G$ together with all regions of the graph which have only one oriented boundary component.
\end{defin}

\begin{remark} \label{rem:walk}
The graph $G$ does not in general form part of the $1$-dimensional subcomplex of $\Sigma_g$. Its vertices represent mobile nodes which may lie anywhere in $\Sigma_g$ and do not necessarily coincide with vertices of the fixed triangulation we have assumed for $\Sigma_g.$ $G$ should for the time being be viewed as an embedded graph on \S which gives rise to the object $\Phi(G)$ in $\mathcal U \subset \R^2$. We do not specify how $\Phi$ will do this. In Section~\ref{sec:model} we mentioned one way in which $\Phi$ could be defined (on a neighborhood of \S in $\R^3$) so straight lines in $\R^3$ between nodes would be mapped by $\Phi$ to PL-curves in the plane. But many choices are possible depending on how one wishes to handle edge-crossings. Therefore we describe GFR in terms of the non-simplicial graph $G$ embedded on \S.

In fact, for the proofs of Section~\ref{sec:proof}, we {\it will} actually subdivide the triangulation of \S so that all curves considered are simplicial. But this is \emph{not done by the algorithm}; it is only a convenience for our proofs. 

Regarding the (non-simplicial) definition of boundary orientation just given, we remark that it assumes the region in question has the same orientation as $\Sigma_g.$ When we speak of $1$-cycles bounding $2$-chains, the latter may be formal sums with coefficients of any sign and $2$-simplices themselves may have orientations that agree or disagree with that of \S.
\end{remark}

In keeping with standard usage, we assume the {\bf right-hand rule} assigns as exit edge at each node
 the next edge clockwise from the incoming edge. Such a routing rule then traverses a border walk containing the starting (directed) edge.

For each reference or connecting curve we assume an initial and final endpoint and orient the curve from the former to the latter. For $\gamma$, these are $S$ and $T$ respectively, for a (closed) reference curve they both correspond to a single point on \S. Moreover, we assume a simple parametrization of each curve: to each point on the curve we associate its {\boldmath{$t$}}{\bf-value}, defined as the distance from the initial point to that point {\it in the direction of orientation}
divided by the total length of the curve. 
This induces a {\bf total ordering} of the crossings of any such curve with the edges of the graph (defining higher crossings as those with higher $t$-values). 

We now define a slower variant of Face Routing. We assume $\alpha$ to be a connecting curve or homology curve and $\beta_0$ to be a border walk that meets $\alpha$ (to start at $S$ with $\alpha = \gamma$, we let $\beta_0$ be the border walk determined by the first exit edge clockwise from $\gamma$ at $S$). In general, we start at a particular crossing $\alpha(t_0)$ of $\alpha$ with $\beta_0$ and {\it for the purposes of the algorithm only}, define {\bf next greater} crossings in terms of $t$-values that have been shifted by subtracting $t_0~{\rm mod}~1$; note: we do not re-define actual $t$-values, just use shifted values to order crossings so as to make $\alpha(t_0)$ the least.

\begin{alg}  [Modified Slow Face Routing (MSFR)]
{During the entire process, stop if $T$ is encountered.} 
\begin{enumerate}
\item $\beta \leftarrow \beta_0$. 
\item Traverse $\beta$ once to find $\alpha(t_{\rm next})$, the \emph{next greater} crossing with $\alpha$ after the entry crossing $\alpha(t_0)$, keeping a record over all of $\beta$ of the sum of crossings to the left minus crossings to the right for each homology curve. 
\item If any of the total sums is non-zero, stop and exit the algorithm, 
otherwise travel to $\alpha(t_{\rm next})$. 
\item Set $\beta$ to the adjacent border walk. $t_0 \leftarrow t_{\rm next}$. 
 Go to (2).
 \end{enumerate}
\end{alg}

\begin{lem}\label{lem:sfr} Let $\alpha$ be a connecting curve or $\alpha \in \mathcal B$.
MSFR fully follows each component of $\alpha\cap {\mathcal R(G)}$
 that it travels unless it meets $T$; more precisely, if $\alpha\cap {\mathcal R(G)}$ starts at $\alpha(t_{0})$ on the NTBW $\beta_1$  and ends at $\alpha(t_{*})$ on the NTBW $\beta_2$, then MSFR along $\alpha$ starting at $\alpha(t_0)$ on $\beta_1$ stops at $\beta_2$ or $T$. And MSFR is reversible (performing MSFR on the oppositely oriented curve starting at $\alpha(t_*)$ on $\beta_2$ will surely reach $\beta_1$). MSFR along all components of $\alpha \cap \mathcal R(G)$ takes a total time of $O(gn^2)$ and memory $O(g\log n)$, $n = |G|$. 
\end{lem}

Although Face Routing finishes in linear time transporting logarithmic memory, it does not have the first properties mentioned (following $\alpha$ within $\mathcal R(G)$ and reversibility). This is addressed in \cite{fra1}. The time bound follows from the bound $d$ on intersections of ``reasonable" curves (since MSFR will travel the boundary of each region in $\mathcal R(G)$ as many times as that walk's edges meet $\alpha$, the walk may have at most $O(n)$ edges and each edge may meet $\alpha$ at most $d$ times); at the same time, for each edge $e$ traversed, MSFR must check all $2g$ curves of $\mathcal B$ for crossings with $e$. The factor of $g$ in the space bound is likewise due to the total sums kept for all these curves (each total can reach be recorded with $O(\log n)$ bits). 

\subsection{Generalized Face Routing (GFR)} \label{subgfr}

Let $\mathcal L$ consist of $\gamma$ together with all elements of $\mathcal B$ and all their oppositely oriented counterparts. 
Assume a fixed {\bf indexing} of these $4g + 1$ elements {\it beginning with $\gamma$}. 

For $\gamma$ and each curve $\alpha \in \mathcal B,$ we have already defined the {\boldmath{$t$}}{\bf-value} of a point on the curve; for the oppositely oriented counterparts of $\alpha \in \mathcal B,$ we define the $t$-value analogously, reversing the roles of $S$ and $T$.

GFR transports a 
list of triples $(i, t_i, t'_i)$, where $i$ is the index of a curve $\alpha$ in $\mathcal L$, and $t_i, t'_i$ are  $t$-values along $\alpha$. 

If MSFR along some $\alpha \in \mathcal L$ stops before reaching $T$, then by Lemma \ref{lem:sfr} it must do so at a NTBW. If at this stage MSFR is begun on another $\alpha' \in \mathcal L$, it will stop at another (possibly same) NTBW. We thus hop from one NTBW to another.

Let $\Gamma = (\mathcal N, \mathcal E)$ be a {\bf virtual multi-graph} whose vertices ($\mathcal N$) are all the NTBW's of the given embedded graph $G$, and whose edges ($\mathcal E$) are connected {\it curve-like}\footnote{If $\alpha$ has {\it no self-intersections} then a connected {\it curve-like} component is just a connected component. If $\alpha$ does have self-intersections, we assume it is parametrized so $\alpha$ is a homeomorphic image of the unit circle. Then retain only that part $\mathcal P$ of the circle mapped into $\mathcal R(G)$. Connected curve-like components of $\alpha \cap \mathcal R(G)$ are images of the connected components of $\mathcal P$.\label{foot:curvelike}} components $c$ of $\alpha \cap \mathcal R(G) : \alpha \in \mathcal B \cup \{\gamma\}$ where $c\subset \alpha$ is an edge between $\beta_1$ and $\beta_2$, for $\beta_1, \beta_2 \in \mathcal N$ if and only if $c$ has its endpoints at these NTBW's. An immediate consequence of this definition is:

\begin{lem} Given an edge $c = (\beta_1, \beta_2)$ of $\Gamma$, which corresponds to a component of $\alpha \in \mathcal B \cup \{\gamma\}$, if MSFR along $\pm\alpha$ is started on $\beta_1$ at the corresponding endpoint of $c$ then it will stop on $\beta_2$. Hence, given any path $\beta_1, \beta_2, \ldots, \beta_k$ in $\Gamma$, such that $T \notin \delta$ for any border walk $\delta$ meeting one of the $c_i = (\beta_i, \beta_{i+1}): i = 1, \ldots, k-1$, reverse MSFR can be used iteratively to return to $\beta_1$ from $\beta_k$ as long as there is available a list of the $t$-values corresponding to the initial and final endpoints of the $c_i$, together with the index of the curve in $\mathcal L$ to which each $c_i$ belongs.
\end{lem}

Using this lemma, we define:
  
\begin{alg}  [Generalized Face Routing (GFR)] \label{alg:gfr}
{During the entire process, exit if $T$ is encountered.}
\begin{enumerate}
\item Use MSFR on $\gamma$ starting from $S$ until it stops. The agent is now at a NTBW, i.e. vertex of $\Gamma$. 
\item Perform a depth first search of $\Gamma$, keeping -- at each subsequent NTBW $\beta$ where MSFR stops -- a record of $(i, t_i, t'_i)$, where $i$ is the index of the curve $\alpha_i \in \mathcal L$ just followed, $t_i$ is the $t$-value along $\alpha$ where MSFR was initiated and $t'_i$ is the $t$-value along $\alpha$ where MSFR reached $\beta$. 
\end{enumerate}
\end{alg}

\begin{remark}
That an NTBW has been previously visited is detected by GFR by the presence of a crossing $\alpha_i(t_i)$ such that $(i, t_i, t'_i)$ and $\alpha_i$ correspond, as described, to one of the recorded triples.
\end{remark}

In the next Section, we will prove:

\begin{prop}\label{lem:2gNTBW} There exists a path in $\Gamma$ from $first(\gamma)$ to $last(\gamma)$, resp. the first and last NTBW's met by $\gamma$. Moreover, $|\Gamma| \leq 2g$. 
\end{prop}

This implies GFR must eventually reach $last(\gamma)$, from which MSFR will be performed on $\gamma$ until $T$ is reached, thus proving the following Theorem. Fig.~\ref{fig:example} illustrates the process.

\begin{figure}[h] 
\centerline{\epsfig{file=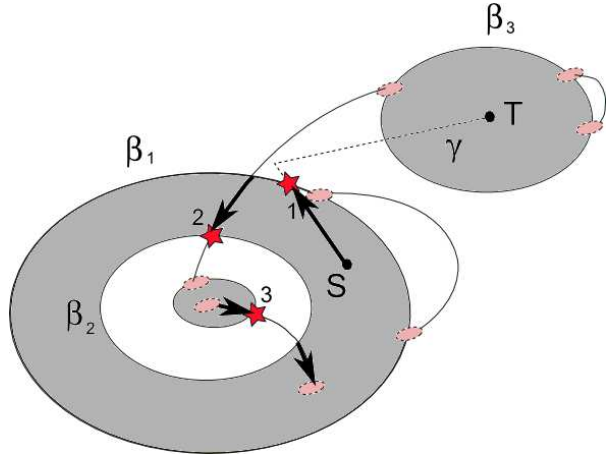,height=6cm}}
\medskip
\caption{This is a planar representation of a surface $\Sigma_g$ (with $g= 4$ in this example). The agent executing GFR sees local portions only of this picture. The shaded area is (the image under $\Phi$ of) $\mathcal R(G)$ - we will omit writing $\Phi$ and keep all notation from $\Sigma$. The curves $\beta_1, \beta_2, \beta_3$ are three of the NTBW's (two more are unlabeled). The dotted piecewise linear curve is $\gamma$; i.e., the connecting curve between $S$ and $T$. The pink disks are the closed disks bounded by $\mu_i$ and $-\bar\mu_i$ respectively; the thin solid curves connecting pink discs in pairs are the $\lambda_i$ ($\mu_i$ paired with $-\bar\mu_i$). GFR on these data follows the arrows in boldface, recording a triple at each star in the numbered sequence. Note that $\beta_3$ appears in two pieces in the plane, but is in reality a single curve (when pink disk pairs are identified). Also, in going from the second star to the third, GFR will first take the arrow into the annular shaded area and then re-emerge in the central disk-like shaded area before arriving at $\beta_3$. From $\beta_3$ the curve $\gamma$ is used to reach $T$.} \label{fig:example}
\end{figure}

\begin{thm} GFR is a local routing algorithm transporting memory of $O(g\log n)$, and guaranteeing delivery in time $O(g^2n^2)$.
\end{thm}

The time bound follows since $|\Gamma| \leq 2g$ and also there are $O(g)$ reference curves in $\mathcal L$ while all edges of $\Gamma$ which correspond to any one of these curves can be traversed by MSFR in time $O(gn^2)$ by Proposition~\ref{lem:2gNTBW}. The record of triples kept by GFR has size $O(g\log n)$ and traversing each edge of $\Gamma$ also requires memory $O(g\log n)$ by the Lemma, hence the space bound.

\section{Proof of Proposition} \label{sec:proof}
We now prove Proposition~\ref{lem:2gNTBW}. It relies on a mathematical fact, Lemma~\ref{lem:connectFaceBdr}, which we prove separately below. From now on we assume that the original triangulation of \S has been subdivided  so that all closed curves referred to -- including those which are NTBW's -- are {\it simplicial}. We do this to prove Proposition only. The algorithm {\it does not do this in practice}.

\begin{proof} We divide the proof into two parts: showing $first(\gamma)$ and $last(\gamma)$ are in the same connected component of $\Gamma$, and establishing the bound on the order of $\Gamma$. 

(1) As we follow $\gamma$ from $first(\gamma)$ to $last(\gamma)$, let $first(\gamma)= \beta_1, \beta_2, \ldots, \beta_{n-1}, \beta_n = last(\gamma)$ be the NTBW's encountered. Successive NTBW's $\beta_i, \beta_{i+1}$ are alternatingly either: \begin{itemize}
\item[-] connected to each other by a component of $\gamma \cap \mathcal R(G)$ or are 
\item[-] boundary components of a common region $F$.
\end{itemize} In the first case, there is by definition an edge between $\beta_i$ and $\beta_{i+1}$ in $\Gamma$. In the second case, by Lemma~\ref{lem:connectFaceBdr}, there is a path in $\Gamma$ between them. Thus we obtain a path
in $\Gamma$ from $first(\gamma)$ to $last(\gamma)$.

(2) To count the NTBW's, suppose there are $N$ of them and let $\mathcal F$ be the collection of all ``non-trivial" regions, meaning regions which have more than one boundary component. Let $k$ be the number of elements in $\mathcal F$. We may cut along all but one NTBW per element of $\mathcal F$ and obtain a surface $\Sigma$ which is still connected, since each region that is an element of $\mathcal F$ will still be attached to $\mathcal R(G)$, which itself is connected. We are cutting along $N - k$ curves,
so $N - k \leq g$. On the other hand, cutting along just one NTBW of each element of $\mathcal F$ would not disconnect the surface either, so $k \leq g$. Hence, $N \leq g + k \leq 2g$.
\end{proof}

\begin{remark} Recall that part (1) of the proof of Proposition~\ref{lem:2gNTBW} was divided into two cases. In the first, we considered vertices of $\Gamma$ (i.e. NTBW's of $G$) which were adjacent in $\Gamma$ since the NTBW's involved were connected by pieces of $\gamma \cap \mathcal R(G)$.
In the second case we appealed to
Lemma~\ref{lem:connectFaceBdr} to establish a path in $\Gamma$ between two NTBW's which share a common region $F$. In fact such a path exists without using any pieces of $\gamma \cap \mathcal R(G)$. In other words if $\Gamma^\ast$ is the subgraph of $\Gamma$ obtained by removing all edges of the form $\gamma \cap \mathcal R(G)$ then the path of interest exists even in $\Gamma^\ast$. We state and prove the Lemma in this stronger form.
\end{remark}

\begin{remark}In the Lemma we will also allow $\mathcal B$ to contain formal sums of simple closed curves, since the added generality simplifies the proof by induction.
\end{remark}

\begin{lem} \label{lem:connectFaceBdr}\footnote{This may be stated more generally, without reference to a graph: given $\mathcal R$, a connected subsurface with boundary of a Riemann surface $\Sigma_g$ with fixed choice of homology basis $\mathcal B$, and two closed curves $\beta, \beta'$ which are boundary components of a single component of $\mathcal R^c$, there exists a sequence of homologically non-trivial curves $\beta = \beta_1, \beta_2, \ldots, \beta_{n-1}, \beta_n = \beta'$ such that each pair of successive curves $\beta_i$ and $\beta_{i+1}$ is connected by a curve segment in the family $\{\alpha \cap \mathcal R~:~\alpha \in \mathcal B\}$.}
Given a set $\mathcal B$ of $1$-cycles with integer coefficients generating $H_1(\Sigma_g; \Q)$
and a connected graph $G$ embedded in $\Sigma_g,$ define $\Gamma^\ast$ to be a virtual multi-graph, with the same vertices as $\Gamma$ in Section~\ref{subgfr} but with edges being connected curve-like\footref{foot:curvelike} components of $\alpha \cap \mathcal R(G) : \alpha \in \mathcal B$ only. Then whenever two NTBW's $\beta$ and $\beta'$ are distinct boundary components of the same region (``face") $F$ determined by $G$, there exists a path in $\Gamma^\ast$ from $\beta$ to $\beta'$.
\end{lem}

\subsection{Intersection Number} In order to prove Lemma~\ref{lem:connectFaceBdr}, we will use certain facts about homology bases expressed in terms of the notion of intersection number. We therefore define intersection number and list its relevant properties, sketching their proof in elementary terms. A formal treatment may be found in \cite{mo}. We recall that \S is a simplicial complex and we use the term {\it closed curve} to mean a connected $1$-cycle (we are using {\it simplicial} homology; see the brief introduction in Section~\ref{sec:hom}). Every closed curve may be written as a sum of simple closed curves, i.e. $1$-cycles with no self-intersections. Now, given any two simple closed curves $\alpha, \beta$, we define their intersection number $\#(\alpha, \beta)$. If $\alpha, \beta$ coincide, we set $\#(\alpha, \beta)$ to zero.
If they do not coincide, let $\sigma$ be a shared connected $1$-chain that is maximal in the sense that at its two distinct endpoints, $q_1$ and $q_2$, $\alpha$ and $\beta$ diverge. We assume that $q_1$ is the start of the shared chain and $q_2$ the end, using the orientation of $\alpha$ (which may not coincide with that of $\beta$). On the other hand, we use the orientation of $\beta$ to establish a left and a right side of $\beta$ (here the orientation of \S is implicit as well). Now define $\iota(\sigma)$ to be zero if $\alpha$ is incident to the same side of $\beta$ at $q_1$ and $q_2$, and define $\iota(\sigma)$ to be $+1$ (resp. $-1$) if $\alpha$ is to the left of $\beta$ at $q_1$ and then to the right at $q_2$ (resp. right and then left). Positive $\iota(\sigma)$ amounts to saying $\alpha$ leaves $\sigma$ as an outward normal to the region bounded by $\beta$. Now, suppose this has been done for all shared sequences $\sigma$ and sum $\iota(\sigma)$ over them to obtain $\#(\alpha, \beta)$. This defines intersection number for all simple closed curves. Extend this definition by bilinearity to {\it formal linear sums} of simple closed curves. The key properties of intersection number that we will use are: 
\begin{enumerate}
\item $\#(\alpha, \beta)$ is a homology invariant, i.e. depends only on the homology classes $[\beta], [\alpha]$
\item $\#(\alpha, \beta)$ is bilinear 
\item $\#(\alpha, \beta)$ is non-degenerate, i.e. any curve which is not homologous to zero must have nonzero intersection number with one of the elements of a homology basis.
\end{enumerate}
These properties show that intersection number provides a non degenerate
bilinear form on the first homology vector space, and they express Poincar\'e
duality in the special case of oriented surfaces (see \cite{ziegler} for an accessible treatment of
Poincar\'e duality and intersection forms).
We outline a proof of the listed properties alone, giving the essential ideas.
\begin{proof-sketch}[of properties above]
We now assume all curves are simple (i.e. have no self-intersections). The general case follows by straightforward extension (property 2. immediately), since arbitrary closed curves are formal linear combinations of simple ones. Moreover, for ease of discussion, we parametrize each curve compatibly with its orientation so we have an order in which to {\it follow} the curve through all its $1$-simplices.

For property 1., note that if $\beta$ and $\tilde\beta$ are homologous then $\beta$ and $-\tilde\beta$ are the two boundary components of some region $F$. As mentioned in the definition, after a shared chain $\sigma$ with positive $\iota(\sigma)$, $\alpha$ exits the region $F$, while after a shared chain with negative $\iota(p)$, $\alpha$ enters $F$. Since any entry (resp. exit) across $\beta$ that does not have a corresponding exit (resp. entry) across $\beta$ must have one across $\tilde\beta$, we conclude that $\#(\alpha, \beta) = \#(\alpha, \tilde\beta)$. The argument in the first coordinate is analogous. Property 2. follows by a similar basic argument since, for example, having $\beta = \beta'+ \beta''$ implies $-\beta$ and $\beta'$ and $\beta''$ are the boundary components of a region $F$. Finally for property 3., note that if $\#(\alpha, \beta) = 0$ for all $\alpha$ in a homology basis $\mathcal B$, then by property 2. $\beta$ would have zero intersection with {\it all} curves of the surface. If $\beta$ were not homologous to zero then removing $\beta$ from \S would not disconnect \S and so there would exist a path within the $1$-skeleton from an edge on one side of $\beta$ to an edge on the other side of $\beta$, not crossing $\beta$ itself. But now, by connecting the two endpoints of this path (to $\beta$ and possibly along it) one would obtain a curve which has intersection number $\pm 1$ with $\beta$, a contradiction. \end{proof-sketch}

\begin{defin}[Crossing] We will refer to a maximal shared connected $1$-chain of $\alpha$ and $\beta$ as a
\emph{crossing} whenever $\alpha$ and $\beta$ are on opposite sides of each other at the start and end of the chain. We will say $\alpha$ \emph{crosses} $\beta$ if such a crossing exists
(even if the intersection number of the two curves is zero).
\end{defin}

\begin{remark} \label{rem:prelem}
Any $1$-cycle $\eta$ which has zero intersection number with a closed curve $\beta$ is homologous to a sum of $1$-cycles which can be parametrized (compatible with orientation) so they do not cross $\beta$. Indeed, one pairs crossings in one direction with crossings in the other and then connects a pair of crossings by two copies of a portion of $\beta$, with suitable (opposite) orientations; when taking the formal sum, these portions of $\beta$ cancel out. 

On the other hand, for any $1$-cycle $\eta$ that has non-zero intersection number with a closed curve $\beta$ there must exist a curve-like\footref{foot:curvelike} connected component $c$ of $\eta \setminus \beta$ which starts on one side of $\beta$ and ends on the other. Indeed if this were not so then following $\eta$ (according to any fixed orientation-compatible parametrization) the values of $\iota$ at the crossings with $\beta$ would alternate between $\pm 1$, thus giving $\#(\eta, \beta) = 0$, a contradiction. 
\end{remark}

\subsection{Proof of Mathematical Lemma}
\begin{proof}[of Lemma~\ref{lem:connectFaceBdr}]
We prove the Lemma by induction on the number of vertices of $\Gamma^\ast$. For $0$ or $1$ vertices it is trivially true. Now suppose $\Gamma^\ast$ has exactly two vertices, which correspond to the NTBW's $\beta$ and $\beta'$. Then these two NTBW's must bound the same region $F$ (since an NTWB by definition is a boundary component of a region with multiple boundary components). 
By property 3. of intersection number there exists $\alpha \in \mathcal B$ such that $\#(\alpha, \beta) \neq 0$. Therefore by Remark~\ref{rem:prelem} there exists at least one component of There is moreover at least one such component $c$. So $\beta$ and $\beta'$ are connected by $c$ and hence $\Gamma^\ast$ is connected.

Now assume the statement of the Lemma holds whenever $\Gamma^\ast$ has $k$ or fewer vertices. Suppose $\Gamma^\ast$ has $k+1$ vertices and $k \geq 2$. Let $\beta$ and $\beta'$ be two distinct NTBW's which bound a common region $F$, and let $\beta''$ be a third NTBW, distinct from $\beta$ and $\beta'$. Now suppose we ``cap off" $\beta''$.
More specifically, by capping off we mean cutting along $\beta''$ and gluing in a small triangulated PLS surface homeomorphic to a closed disk on either side of the cut. The resulting surface $\Sigma'$ has genus $g-1$ and it remains connected since $G$ is (moreover, it is PLS and triangulated). The caps we glue in to $\beta''$ and $-\beta''$ are both closed disks and we will retain the name $\beta''$ for the boundary of the disk capping off $\beta''$. Moreover, since $\beta''$ was a subgraph of $G$ -- in fact the only part of $G$ that was affected by the cutting and gluing -- let the new $\beta''$ still form the same part of $G$ and continue to use the name $G$ for the resulting embedded graph on $\Sigma'$. Notice that for this graph the tiled region will be extended by exactly one disk and $\beta''$ will no longer be a NTBW but all other NTBW's in \S will still be NTBW's in $\Sigma'.$ This is because even if $\beta''$ in \S was a boundary component of $F$, there would still remain the two boundary components $\beta$ and $\beta'$ for the region after capping off and by definition of NTBW, both $\beta$ and $\beta'$ would still be NTBW's. If we denote the {\bf new tiled region} in $\Sigma'$ by $\mathcal R(G)'$, we have $\mathcal R(G) \subset \mathcal R(G)'$. 

By a standard homological argument (given in Lemma~\ref{lem:basis} below) there exists a a set $\mathcal B'$ of $1$-cycles with integer coefficients generating $H_1(\Sigma'; \Q)$ such that for each $\alpha \in \mathcal B'$, every connected curve-like component of 
$\alpha \cap \mathcal R(G)'$ is either a connected curve-like component of 
$\eta \cap \mathcal R(G)$
for some $\eta \in \mathcal B$, or else a union of portions of $\beta''$ together with  connected curve-like components of $\eta \cap \mathcal R(G)$ that meet $\beta''$. If we define $(\Gamma^\ast)'$ to be the corresponding virtual multigraph, with vertices that are NTBW's in $\Sigma'$ and edges that are connected curve-like components of $\alpha \cap \mathcal R(G)'$ for $\alpha \in \mathcal B'$, then we see that the vertex set of $(\Gamma^\ast)'$ is obtained from that of $\Gamma^\ast$ by eliminating the single vertex $\beta''$ (and no other). Moreover (using this correspondence of vertices) the property of $\mathcal B'$ given above means that
any edge of $(\Gamma^\ast)'$ (from $\nu$ to $\nu'$ say) is either an edge of $\Gamma^\ast$ or one can find two edges in $\Gamma^\ast$: one from $\nu$ to $\beta''$, the other from $\beta''$ to $\nu'$. Thus a path between vertices in $(\Gamma^\ast)'$ implies a path between those same vertices in $\Gamma^\ast$.

We are now done since we know by the inductive hypothesis that there is a path connecting $\beta$ and $\beta'$ in $(\Gamma^\ast)'.$
\end{proof}

\begin{figure}[h] 
\centerline{\epsfig{file=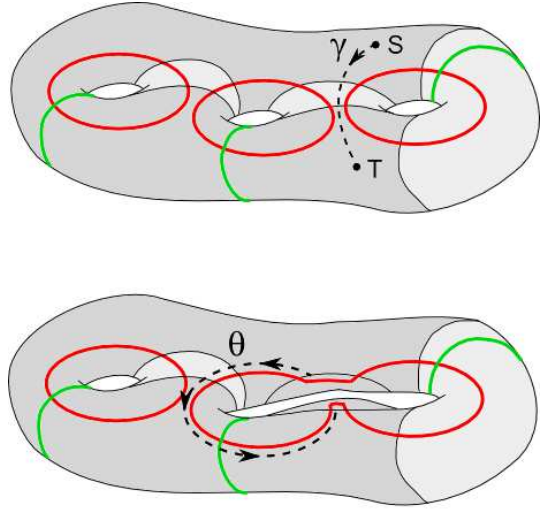,height=7cm}}
\medskip
\caption{Induction on the genus. The top surface is $\Sigma_g$, the bottom $\Sigma'$ obtained by cutting along the NTBW $first(\gamma)$ and gluing a disk in on each side. The darkly shaded area represents the tiled region in each surface. The red and green curves in the top surface are its $\lambda_i$ and $\mu_i$ respectively (same surface as in Figure~\ref{fig:handles}). The red curve in the middle plays the role of $\lambda$, by means of which we define the new reference curves for $\Sigma'$ as shown on the bottom surface (red/green retained). If we solve the routing problem based on a connecting curve $\theta$ in the surface $\Sigma'$ (which has been taken to be $\lambda\setminus first(\gamma)$ and is shown dashed) this will give us a solution to the original problem which was in terms of $\gamma$ in $\Sigma_g$. But the surface $\Sigma'$ has genus strictly less than $g$, so we are done by induction.} \label{fig:induction}
\end{figure}

The argument we just made is illustrated in Fig.~\ref{fig:induction} in terms of routing.
We now state and prove the result we used concerning homology bases and capping off.

\begin{lem} \label{lem:basis}
Let $\mathcal B$ be a set of $1$-cycles with integer coefficients which generates $H_1(\Sigma_g; \Q)$ and $\Sigma'$ the surface obtained by capping off a closed homologically non-trivial curve $\beta$ (and $-\beta$ too). Then there exists another set $\mathcal B'$ of $1$-cycles with integer coefficients which generates $H_1(\Sigma'; \Q)$ such that for each $\alpha \in \mathcal B'$, every connected curve-like\footref{foot:curvelike} component of 
$\alpha \cap \mathcal R(G)'$ is either a connected curve-like component of 
$\eta \cap \mathcal R(G)$
for some $\eta \in \mathcal B$, or else a union of portions of $\beta$ together with connected curve-like components of $\eta \cap \mathcal R(G)$ that meet $\beta$. 
\end{lem}

\begin{proof}
We construct a set of $1$-cycles which generates the homology of $\Sigma'$ and has the required properties.

Since $\beta$ is not homologous to zero, there is some $\alpha \in \mathcal B$ such that $\#(\alpha, \beta) \neq 0$ by property 3. above. Denote it $\lambda$ and let $m = \#(\lambda, \beta) \neq 0$. Now consider the following set of $1$-cycles with integer coefficients:
$$\mathcal S = \{m \alpha - \#(\alpha, \beta)\lambda ~:~ \alpha \in \mathcal B\}.$$
All of these curves have zero intersection number with $\beta$. Using Remark~\ref{rem:prelem}, each one is homologous to a sum of (parametrized) closed curves not crossing $\beta$. Let $\mathcal S'$ be the set of all these closed curves; they are also closed curves in $\Sigma'$. We claim the elements of $\mathcal S'$ generate the first homology of $\Sigma'$. 

Warning: when we write $[\cdot]$ we mean a homology class in $H_1(\Sigma_g, \Q)$; the homology classes in $H_1(\Sigma_g, \Q)$ we will simply describe as such and not give them a short-hand notation.

Let $\delta$ be an arbitrary closed curve in $\Sigma'$. Since $\beta$ bounds an open region $D$ in $\Sigma'$ that is homeomorphic to an open disk, any piece of $\delta$ which is a chord of $D$ is homologous (in $\Sigma'$) to a connected sub-chain of $\beta$. Using this to substitute all such chords of $\delta$, we obtain that $\delta$ is homologous (in $\Sigma'$) to a curve which does not cross $\beta$ and in fact lies outside $D$. Without loss of generality we may take this new curve as $\delta$ (since we are interested in generating the homology of $\Sigma'$). It corresponds to a curve in \S (by replacing $1$-simplices on the boundary of the cap with $1$-simplices of $\beta$, as usual). Denote this curve in \S also by $\delta$. It is homologous in \S to a linear combination of the elements of $\mathcal B$, 
say 
$$[\delta] = \sum\limits_{\alpha \in \mathcal B} a(\alpha) [\alpha]$$
with coefficients $a(\alpha) \in \Q$ indexed by $\alpha \in \mathcal B$. We have
$$0 = \#(\delta, \beta) = \sum\limits_{\alpha \in \mathcal B} a(\alpha) \#(\alpha, \beta).$$
Therefore,
\begin{eqnarray*}
\left [ \sum\limits_{\alpha \in \mathcal B} \frac{a(\alpha)}{m} 
  \left \{ m\alpha- \#(\alpha, \beta)\lambda \right \} \right ]
&=& \sum\limits_{\alpha \in \mathcal B} \frac{a(\alpha)}{m} [m\alpha] 
  - \frac{1}{m}\sum\limits_{\alpha \in \mathcal B} a(\alpha)\#(\alpha, \beta) [\lambda] \\
&=& \sum\limits_{\alpha \in \mathcal B} a(\alpha) [\alpha] = [\delta].
\end{eqnarray*}
None of the elements of $\mathcal S'$ crosses $\beta$ and so neither does a linear combination of them. We thus have two curves, $\sum\limits_{\alpha \in \mathcal B} \frac{a(\alpha)}{m} 
  \left \{ m\alpha- \#(\alpha, \beta)\lambda \right \}$ and $\delta$ which lie outside $D$ and are homologous to each other in $\Sigma_g.$ Denote the first curve $\eta$ and let $\kappa$ be a $2$-chain in \S with boundary $\eta-\delta$.
When we cap off $\beta$ as described earlier we produce the curve $\eta-\delta$ in $\Sigma$ and a new $2$-chain $\kappa'$ in $\Sigma_g'$ whose boundary is $\eta-\delta$.  Thus $\eta$ and $\delta$ are also homologous in $\Sigma'.$
So $\mathcal S'$ indeed generates the first homology of $\Sigma'$. 
\end{proof}

\section{Conclusions} Face Routing (FR) is at the base of most position-based routing algorithms used in ad hoc networks today. It is a logarithmic-memory, local algorithm guaranteeing delivery in embedded graphs on the plane but not on positive genus surfaces. We have exhibited a position-based algorithm, GFR, which guarantees delivery for embedded graphs on surfaces of arbitrary genus. It is also local and uses logarithmic memory. Like FR, it is a base algorithm which as a next step could be adapted to handle edge-crossings. We welcome communication on particular settings of interest, as the method of handling edge-crossings would have to be somewhat tailored to the application. While universal exploration sequences (UXS's) provide a {\it non position-based} 
logarithmic-memory, local routing algorithm for any network, the time is a high polynomial ($> 16$) and no better algorithm has until now been proposed which would apply to 3D ad hoc networks. FR and GFR take linear and quadratic time respectively. By restricting to the setting of graphs on surfaces which retains the codimension-one character of graphs in the plane, we have obtained (with GFR) an FR-like gain in time efficiency for the task of local, logarithmic-memory routing in a class of 3D ad hoc network.


\begin{thebibliography}{99}



\bibitem{kranurr2} P. Boone, E. Chavez, L. Gleitzky, E. Kranakis, 
J. Opatrny, G. Salazar and J. Urrutia, 
Morelia Test: Improving the Efficiency of the Gabriel Test and Face Routing in Ad-hoc Networks, 
{in Proc. SIROCCO 2004}, R. Kralovic and O. Sykora eds., Springer Verlag, 2004,
LNCS {3104}, pp. 23-34.

\bibitem{braverm} M. Braverman, {On ad hoc routing with guaranteed delivery}, arxiv (2008),
{\tt http://arxiv.org/abs/0804.0862}.

\bibitem{lata}
S. Durocher, D. Kirkpatrick and L. Narayanan, 
{On Routing with Guaranteed Delivery in Three-Dimensional Ad Hoc Wireless Networks},
{in Proc. ICDCN 2008}, Springer Verlag, 2008
LNCS {4904}, pp. 546-557.

\bibitem{erickson}
J. Erickson and K. Whittlesey,
{Greedy optimal homotopy and homology generators},
{in Proc. SODA 2005}, SIAM, 2005, pp. 1038-1046. 

\bibitem{fra1} M. Fraser, {Local Routing on Tori}, {in Proc. ADHOC-NOW 2007}, 
Springer Verlag, 2007,
LNCS 4686, pp. 153-166.

\bibitem{fraJoint} M. Fraser, E. Kranakis J. Urrutia, {Memory Requirements for 
Local Geometric Routing and Traversal in Digraphs}, in
Proc. CCCG, 2008, pp. 195-198. 

\bibitem{hatcher} A. Hatcher, {Algebraic Topology}, Cambridge University Press, 2001.

\bibitem{mo} M. Hirsch, {Differential Topology}, 6th ed., Springer Verlag, 1997.

\bibitem{klein} R. Klein, {Linear Approximation of Trimmed Surfaces}, {in The Mathematics Of Surfaces VI}, Oxford Universitry Press, 1994, pp. 209-212. 

\bibitem{kranurr} E. Kranakis, H. Singh and J. Urrutia, {Compass Routing on Geometric Networks}, 
in Proc. CCCG, 1999, pp. 51-54.

\bibitem{kuhn} F. Kuhn, R. Wattenhofer and A. Zollinger, {Ad-Hoc Networks 
Beyond Unit Disk Graphs}, {in Proc. DIALM-POMC 2003}, ACM, 2003, pp. 69-78.

\bibitem{margaliot} M. Margaliot, C. Gotsman, {Piecewise-linear surface approximation from noisy scattered samples}, in 
Proc. IEEE conference on Visualization, IEEE, 1994, pp. 61-68.

\bibitem{massey} W.S. Massey, {Algebraic Topology: An Introduction}, Springer Verlag, 1967.

\bibitem{stconn} O. Reingold, {Undirected ST-Connectivity in Log-Space}, {in
Proc. STOC 2005}, ACM, 2005, pp. 376-385.

\bibitem{vegter} G. Vegter, {Computational Topology}, {in CRC Handbook of Discrete and Computational Geometry}, J.E. Goodman and J. O'Rourke eds., CRC Press, 2004, Chapter 32, pp. 719-742.

\bibitem{yamada} A. Yamada, K. Shimada and T. Itoh, {Energy-Minimizing Approach to Meshing Curved Wire-Frame Models}, from 5th International Meshing Roundtable, Sandia National Laboratories, 1996, pp.179-194.

\bibitem{ziegler} G. M. Ziegler, online course notes, Topology BMS Basic Course � Winter 07/08, TU Berlin, 2008, {\tt  http: //www.math.tu-berlin.de/\~{}ziegler/TOP3/notes10.pdf.}

\end{thebibliography}
\end{document}